%% file: WWW_2019_Block.tex
\pdfoutput=1
\documentclass[sigconf,nonacm]{acmart}

\usepackage{booktabs} 

\usepackage{subfig}
\usepackage{algorithm}
\usepackage[noend]{algorithmic}

\usepackage{amsmath}
\usepackage{amsthm}
\theoremstyle{plain}

\newtheorem*{theorem*}{Theorem}
\usepackage{multirow}

\usepackage{natbib}

\usepackage{balance}

\usepackage[inline,nomargin]{fixme}

\DeclareMathOperator{\poly}{poly}

\DeclareMathOperator{\E}{\mathbb{E}}

\renewcommand{\vec}[1]{\boldsymbol{\mathbf{#1}}}

\usepackage{hyperref}

\begin{document}

\title{The Block Point Process Model for Continuous-time Event-based Dynamic Networks}
\author{Ruthwik R. Junuthula}
\affiliation{%
 \institution{University of Toledo}
 \streetaddress{2801 W. Bancroft St.}
 \city{Toledo}
 \state{OH}
 \country{USA}
 \postcode{43606}
}
\email{rjunuth@utoledo.edu}

\author{Maysam Haghdan}
\authornote{Research was conducted while author was at the University of Toledo.}
\affiliation{%
 \institution{ETH Z\"{u}rich}
 \streetaddress{Universit\"{a}tstrasse 25}
 \city{Z\"{u}rich}
 \country{Switzerland}
 \postcode{8092}
}
\email{maysam.haghdan@inf.ethz.ch}

\author{Kevin S. Xu}
\affiliation{%
 \institution{University of Toledo}
 \streetaddress{2801 W. Bancroft St.}
 \city{Toledo}
 \state{OH}
 \country{USA}
 \postcode{43606}
}
\email{kevin.xu@utoledo.edu}

\author{Vijay K. Devabhaktuni}
\affiliation{%
 \institution{University of Toledo}
 \streetaddress{2801 W. Bancroft St.}
 \city{Toledo}
 \state{OH}
 \country{USA}
 \postcode{43606}
}
\email{Vijay.Devabhaktuni@utoledo.edu}

\begin{abstract}
We consider the problem of analyzing timestamped relational events between a set of entities, such as messages between users of an on-line social network. 
Such data are often analyzed using static or discrete-time network models, which discard a significant amount of information by aggregating events over time to form network snapshots. 
In this paper, we introduce a block point process model (BPPM) for 
continuous-time event-based dynamic networks. 
The BPPM is inspired by the well-known stochastic block model (SBM) for static 
networks. 
We show that networks generated by the BPPM follow an SBM in the limit of a 
growing number of nodes. 
We use this property to develop principled and efficient local search and variational 
inference procedures initialized by regularized spectral clustering. 
We fit BPPMs with exponential Hawkes processes to analyze several real network data sets, including a 
Facebook wall post network with over $3,500$ nodes and $130,000$ events.  
\end{abstract}



\maketitle

\input{intro.tex}

\input{background.tex}

\input{model.tex}

\input{inference.tex}

\input{experiments.tex}

\input{conclusion.tex}

\appendix

\input{proofDeviation.tex}

\input{appendix.tex}

\begin{acks}
This material is based upon work supported by  the \grantsponsor{nsf}{National Science Foundation}{https://nsf.gov} grants \grantnum{nsf}{IIS-1755824} and \grantnum{nsf}{DMS-1830412}.
\end{acks}

\bibliographystyle{ACM-Reference-Format}
\balance
\bibliography{WWW_2019_Block}

\end{document}

%% file: intro.tex
\section{Introduction}

Many application settings involve analysis of timestamped relational event 
data in the form of triplets (sender, receiver, timestamp), as shown 
in Figure \ref{fig:EventTable}. 
Examples include analysis of messages between users of an on-line social 
network, emails between employees of a company, and transactions between 
buyers and sellers on e-commerce websites. 
These types of data can be represented as dynamic networks evolving in 
continuous time due to the fine granularity on the timestamps of 
events and the irregular time intervals at which events occur.

\begin{figure}[t]
\centering
\subfloat[Event table]{
\label{fig:EventTable}
\begin{tabular}[b]{ccc}
\hline
Sender & Receiver & Time \\
\hline
1 & 2 & 0.1 \\
2 & 3 & 0.4 \\
3 & 2 & 0.6 \\
1 & 2 & 1.2 \\
1 & 3 & 1.3 \\
2 & 1 & 1.6 \\
\hline
\end{tabular}
}
\qquad
\subfloat[Adjacency matrix]{
\label{fig:AdjMats}
\begin{minipage}[b]{1.2in}
\begin{align*}
A^{[0,1)} &= 
\begin{bmatrix}
0 & 1 & 0 \\
0 & 0 & 1 \\
0 & 1 & 0
\end{bmatrix} \\
A^{[1,2)} &=
\begin{bmatrix}
0 & 1 & 1 \\
1 & 0 & 0 \\
0 & 0 & 0
\end{bmatrix}
\end{align*}
\vspace{1pt}
\end{minipage}
}
\caption{Two representations of a continuous-time event-based dynamic 
network.
The adjacency matrices discard the exact times and ordering of 
events.}
\label{fig:EventVsAdj}
\end{figure}

Statistically modeling these types of relations and their dynamics over time 
has been of great interest, especially given the ubiquity 
of such data in recent years. 
Most prior work has involved modeling these relations using 
network representations, with nodes representing senders and 
receivers, and edges representing events. 
Such representations often either discard the timestamps 
altogether, which transforms the dynamic network into a static network, or 
aggregate events over time windows to form network snapshots evolving in 
discrete time as in Figure \ref{fig:AdjMats}.
There have been numerous statistical models proposed for static networks 
dating back to the 1960s \citep{Holland1983,goldenberg2010survey}, and more recently, for 
discrete-time networks \citep{xing10,ho11,Yang2011,Xu2014a,Xu2015,Han2014,Corneli2015,Matias2016,rastelli2017exact}, but 
comparatively less attention has been devoted to 
continuous-time networks of timestamped relations.

The development of such continuous-time or point process-based network models \citep{DuBois2010,Blundell2012,Dubois2013,xin2015continuous,
matias2015semiparametric,Fox2016} appears to have progressed separately from recent advances in static and discrete-time network models. 
There have been many recent developments on estimation for static network models such as the stochastic block model (SBM), including the development of consistent estimators such as regularized spectral clustering \cite{qin13,rohe2016}. 
Although some continuous-time models have drawn their inspiration from static network models, to the best of our knowledge, there has not been prior work connecting the two types of models, and in particular, examining whether provably accurate estimators for static network models can be used to estimate continuous-time models. 

In this paper we introduce the block point process model (BPPM) for 
continuous-time event-based 
dynamic networks, inspired by the SBM for static networks. 
Our main contributions are as follows: 
\begin{itemize}
\item We demonstrate an asymptotic equivalence between our proposed BPPM and the SBM in the limit of growing number of nodes, which allows us to use provably accurate and efficient estimators for the SBM, such as regularized spectral clustering, as a starting point to fit BPPMs.

\item We develop efficient local search and variational inference procedures for the BPPM initialized by regularized spectral clustering on an aggregated adjacency matrix.

\item We fit the BPPM to several real network data sets, including a Facebook network with over $3,500$ nodes and $130,000$ events and demonstrate that it is more accurate at predicting future interactions compared to discrete-time SBMs.
\end{itemize}

%% file: background.tex
\section{Background}
\label{sec:Background}

We consider dynamic networks evolving in continuous time through the 
observation of events between pairs of nodes at recorded timestamps, as shown 
in Figure \ref{fig:EventTable}. 
We assume that events are directed, so we refer to the two nodes involved 
in an event as the sender and receiver
(although the model we propose can be trivially 
modified to handle undirected events by reducing the number of parameters). 
Such event data can be represented in the form of a matrix $E$ where each 
row is a 
triplet $\vec{e} = (u,v,t)$ denoting 
an event from node $u$ to node $v$ at timestamp $t$. 
Let $N$ denote the total number of nodes in the network, and let $T$ denote 
the time of the last interaction, so that the interaction times are all in 
$[0,T]$. 

From an event matrix $E$, one can obtain an adjacency matrix 
$A^{[t_1,t_2)}$ over any given time interval $[t_1,t_2)$ such 
that $0 \leq t_1 < t_2 \leq T$. 
To simplify notation, we drop the time interval from the adjacency matrix, 
i.e.~$A = A^{[t_1,t_2)}$.
In this adjacency matrix, $a_{ij} = 1$ if there is at least one 
event from node $i$ to node $j$ in $[t_1,t_2)$, and $a_{ij} = 0$ 
otherwise. 
For example, Figure \ref{fig:AdjMats} shows two adjacency matrices 
constructed by aggregating events from the event table shown in Figure \ref{fig:EventTable} over $[0,1)$ and $[1,2)$. 

\subsection{The Stochastic Block Model}
\label{sec:AdjModels}
Most statistical models for networks 
consider an adjacency matrix rather than event-based representation; many 
commonly used models of this type are discussed in the survey by 
\citet{goldenberg2010survey}. 
One model that has received significant attention 
is the \emph{stochastic block model} (SBM), which is defined as 
follows (adapted from Definition 3 in \citet{Holland1983}):

\begin{definition}[Stochastic block model]
\label{def:SBM}
Let $A$ denote a random adjacency matrix for a static network, and let 
$\vec{c}$ denote a class membership vector. 
$A$ is generated according to a stochastic block model with respect 
to the membership vector $\vec{c}$ if and only if,

\begin{enumerate}
\item For any nodes $i \neq j$, the random variables $a_{ij}$ are 
statistically independent.
\label{item:SBM_Indep}

\item For any nodes $i \neq j$ and $i' \neq j'$, if 
$i$ and $i'$ are in the same class, i.e.~$c_i = c_{i'}$,
and $j$ and $j'$ are in the same class, i.e.~$c_j = c_{j'}$,
then $a_{ij}$ and $a_{i'j'}$ are identically distributed.
\label{item:SBM_Ident}
\end{enumerate}

\end{definition}

The classes in the SBM are also commonly referred to in the literature as 
blocks. 
The class membership vector $\vec{c}$ has $N$ entries where each entry 
$c_i \in \{1, \ldots, K\}$ denotes the class membership of node $i$, 
and $K$ denotes the total number of classes.
Recent work has focused on estimating the class memberships 
from the adjacency matrix $A$. 
In this setting, spectral clustering (and regularized variants) has emerged as an efficient estimator 
that has theoretical accuracy guarantees 
\citep{Rohe2011,Sussman2012,qin13,Lei2015,rohe2016}, 
scales to large networks with thousands of nodes, and is generally not sensitive to 
initialization.

\subsection{Related Work}
\label{sec:PointProcess}

Most existing work on modeling dynamic networks has considered a 
discrete-time representation, where the observations consist of a sequence 
of adjacency matrices. 
This observation model is ideally suited for network data collected at 
regular time intervals, e.g.~weekly surveys. 
In practice, however, dynamic network data is often collected at much finer 
levels of temporal resolution (e.g.~at the level of a second or 
millisecond), in which case 
it likely makes more sense to treat time as continuous rather than discrete. 
In order to apply discrete-time dynamic network models to such data, it 
must first be pre-processed by aggregating events over time windows to form 
network snapshots, and this technique is used in many real data experiments 
\citep{Yang2011,Xu2014a,Han2014,Xu2015,matias2015semiparametric}. 
For example, an aggregated representation of the network in Figure 
\ref{fig:EventTable} with time window of $1$ is shown in Figure 
\ref{fig:AdjMats}.

Aggregating continuous-time network data into discrete-time snapshots 
presents several challenges. 
One would ideally choose the time window to be as short as possible for the 
maximum temporal resolution. 
However, this increases the number of snapshots, and accordingly,  
the computation time (typically linear in the number of snapshots). 
More importantly, models fit using shorter time windows can lead to worse 
predictors than models fit using longer time windows because the models 
often assume short-term memory, such as the Markovian dynamics in many 
discrete-time SBMs \citep{xing10,ho11, Yang2011,Xu2014a,Xu2015,Matias2016}. 
We demonstrate some of these practical challenges in an experiment in 
Section \ref{sec:ExpDisc}.

Another line of research that has evolved independently of discrete-time 
network models involves the use of point processes to 
estimate the structure of a latent network from observations at 
the nodes 
\citep{hall2014tracking,linderman2015scalable,Farajtabar2015,Tran2015,He2015}. 
These models are often used to estimate networks of diffusion from information 
cascades. 
Such work differs from the setting we consider in this paper, where we 
directly observe events \emph{between pairs of nodes} and seek to model the 
dynamics of such event sequences. 

There have been several other models proposed using point processes to model 
continuous-time event-based networks
\citep{DuBois2010,Blundell2012,Dubois2013,Linderman2014,xin2015continuous,
matias2015semiparametric,Fox2016}, which is the setting we consider in this 
paper. 
These models are typically fit using Markov chain Monte Carlo (MCMC) methods, which do not scale to large 
networks with thousands of nodes. 
The BPPM that we propose in this paper is a simpler version of the Hawkes IRM 
\citep{Blundell2012}. 
The relational event model (REM) \citep{Dubois2013} is related to the 
BPPM in that it is also inspired by the SBM and 
shares parameters across nodes in the network in a similar manner. 
We discuss the Hawkes IRM and REM in greater detail and compare them to our 
proposed model in Section \ref{sec:OtherModels}.

%% file: model.tex
\section{The Block Point Process Model}
\label{sec:Model}

\subsection{Model Specification}
\label{sec:ModelSpec}
We propose to model continuous-time dynamic networks using a generative 
point process network model. 
Motivated by the SBM for static networks, 
we propose to divide nodes into $K$ classes or blocks and to associate a 
univariate point process with each pair of node blocks 
$b = (b_1,b_2) \in \{1, \ldots, K\}^2$, which we refer to as a 
\emph{block pair}. 
Let $p=K^2$ denote the total number of block pairs. 
Let $\vec{\pi} = \{\pi_1, \ldots, \pi_K\}$ denote the class membership 
probability vector, where $\pi_q$ denotes the 
probability that a node belongs to class $q$. 
We call our model the \emph{block point process model} (BPPM). 
The generative process for the BPPM for a network of duration $T$ time units 
is shown in Algorithm \ref{alg:GenBPPM}. 

\begin{algorithm}[t]
\caption{Generative process for BPPM}
\label{alg:GenBPPM}

\begin{algorithmic}[1]
\FOR{node $i=1$ \TO $N$}
	\STATE Sample class $c_i$ from categorical distribution with parameter 
		vector $\vec{\pi}$
\ENDFOR
\FOR{block pair $b=1$ \TO $p$}
	\LOOP
		\STATE Sample next event time $t_b$ from $b$th point process
		\label{item:genEvent}
		\IF{$t_b > T$}
			\STATE \textbf{break}
		\ENDIF
		\STATE Randomly select nodes $i \in b_1, j \in b_2$ to form an edge 
			from node $i$ to node $j$ at time $t_b$
		\label{item:selectNodes}
	\ENDLOOP
\ENDFOR
\end{algorithmic}
\end{algorithm}

The BPPM is a very general model---notice that we have not specified what 
type of point process to use in the model (we discuss this in Section 
\ref{sec:PointProChoice}). 
The proposed BPPM is less flexible than existing point process network models such as the Hawkes IRM and the REM (we compare the BPPM to these models in Section \ref{sec:OtherModels}), 
but its simplicity enables theoretical analysis of the model. 
We then use the findings of our analysis to develop principled
and efficient inference procedures that scale to large networks with 
thousands of nodes and hundreds of thousands of events. 
The proposed inference procedures, which we discuss in Section 
\ref{sec:Inference}, take advantage of the close relationship between the 
BPPM and the SBM, which we discuss next.

\subsection{Asymptotic Equivalence with the Stochastic Block Model}
\label{sec:RelationSBM}

The BPPM is motivated by the SBM, where the probability of forming an edge 
between two nodes depends only the classes of the two nodes. 
Given the relation between the point process and adjacency matrix 
representations discussed in Section \ref{sec:Background}, a natural 
question is whether there is any equivalence between the BPPM and the SBM. 
Specifically, does an adjacency matrix $A = A^{[t_1,t_2)}$ constructed from 
an event matrix $E$ generated by the BPPM follow an SBM? 
As far as we know, this connection between point process 
and static network models has not been previously explored in the literature. 

We first note that $A$ meets criterion \ref{item:SBM_Ident} (identical 
distribution within a block pair) in Definition \ref{def:SBM} due to the 
random 
selection of node pair for each event in step \ref{item:selectNodes} of 
Algorithm \ref{alg:GenBPPM}. 
To check criterion \ref{item:SBM_Indep} (independence of all entries of 
$A$), we first note that entries $a_{ij}$ and $a_{i'j'}$ in 
different block pairs, i.e.~$(c_i,c_j) \neq (c_{i'},c_{j}')$, depend on 
different independent point processes (unlike in the Hawkes IRM), so $a_{ij}$ and $a_{i'j'}$ are 
independent. 

Next, consider entries $a_{ij}$ and $a_{i'j'}$ in the same block pair 
$b = (c_i,c_j) = (c_{i'}, c_{j'})$. 
In general, these entries are dependent so 
that criterion \ref{item:SBM_Indep} is not satisfied.\footnote{An 
exception is the case of a homogeneous Poisson process, for which the 
entries are independent by the splitting property.} 
For example, if a Hawkes process \citep{Laub2015} is used in step \ref{item:genEvent} 
of Algorithm \ref{alg:GenBPPM}, then 
$a_{i'j'} = 1$ indicates that at least one event was 
generated in block pair $b$, i.e.~there was at least one jump in the intensity 
of the process. 
This indicates that the probability of 
another event is now higher, so the conditional probability 
$\Pr(a_{ij}=1 | a_{i'j'} = 1)$ should be higher than the marginal 
probability $\Pr(a_{ij}=1)$. 
Thus $a_{ij}$ and $a_{i'j'}$ are 
\emph{dependent}, so $A$ does \emph{not follow} an SBM! 

We denote the \emph{deviation from independence} using the terms $\delta_0$ 
and $\delta_1$ defined by
\begin{align}
\label{eq:delta0}
\delta_0 &= \Pr(a_{ij}=0 | a_{i'j'}=0) - \Pr(a_{ij}=0) \\
\label{eq:delta1}
\delta_1 &= \Pr(a_{ij}=0 | a_{i'j'}=1) - \Pr(a_{ij}=0).
\end{align}
If $\delta_0 = \delta_1 = 0$, then the two adjacency matrix entries are 
independent.   
If $\delta_0 \neq 0$ or $\delta_1 \neq 0$, then the two entries are dependent, 
with smaller values of 
$|\delta_0|,|\delta_1|$ indicating less dependence. 
The following theorem bounds these values.

\begin{theorem*}[Asymptotic Independence Theorem]
\label{thm:AsyIndep}
Consider an adjacency matrix $A$ constructed 
from the BPPM over some time interval $[t_1,t_2)$. 
Then, for any two entries $a_{ij}$ and $a_{i'j'}$ both in block pair 
$b$, the deviation from independence given by $\delta_0, \delta_1$ 
defined in \eqref{eq:delta0}, \eqref{eq:delta1} is bounded in the following manner:
\begin{equation}
\label{eq:deltaBound}
|\delta_0|,|\delta_1| \leq \min\left\{1,\mu_b/n_b\right\}
\end{equation}
where $\mu_b$ denotes the expected number of events in block pair $b$ in 
$[t_1,t_2)$, and $n_b$ denotes the size of block pair $b$. 
In the limit as $n_b \rightarrow \infty$, 
$\delta_0,\delta_1 \rightarrow 0$ provided $\mu_b$ grows at a slower rate 
than $n_b$. 
Thus $a_{ij}$ and $a_{i'j'}$ are asymptotically independent for growing $n_b$. 
\end{theorem*}
The proof of the \nameref{thm:AsyIndep} is provided in Appendix 
\ref{sec:proofDeviation}. 
We evaluate the tightness of the bound in \eqref{eq:deltaBound} via simulation 
in Section \ref{sec:DevExp}. 
Since it depends only on the expected number of events $\mu_b$ and not the 
distribution, it is likely to be loose in general but applies to any choice of point process. 

The \nameref{thm:AsyIndep} states that the deviation given by 
$\delta_0,\delta_1$ is non-zero in 
general for fixed $n_b$, so the entries $a_{ij}$ and $a_{i'j'}$ are dependent, 
but the dependence decreases as the size of a block (and thus, a block pair) 
grows. 
This can be achieved by letting the number of nodes $N$ in the network 
grow while holding the number of classes $K$ fixed. 
In this case, the sizes of block pairs would be growing at rate $O(N^2)$, 
so the asymptotic behavior should be visible for networks with thousands of 
nodes. 
Thus, an adjacency matrix constructed from the BPPM approaches an SBM in 
the limit of a growing network! 
To the best of our knowledge, this is the first such result linking 
networks constructed from point 
process models and static network models. 
It is also practically useful in that it allows us to leverage recent work 
on provably accurate and efficient inference on the SBM for the BPPM.

\subsection{Choice of Point Process Model}
\label{sec:PointProChoice}
Any temporal point process can be used to generate the event times 
in the BPPM. 
We turn our attention to a specific point process: the Hawkes 
process  \citep{Laub2015}, which is a self-exciting process where the 
occurrence of events 
increases the probability of additional events in the future. 
The self-exciting property tends to create clusters of events in time, which 
are empirically observed in many settings.
Prior work has suggested 
that Hawkes processes with exponential kernels provide a good fit to many 
real social network data sets, including email and conversation sequences 
\cite{halpin2013modelling,masuda2013self} and re-shares 
of posts on Twitter \cite{zhao2015seismic}. 
Hence, we also adopt the exponential kernel, which has intensity 
function
\begin{equation*}
\lambda(t) = \lambda^\infty + \sum_{t_i < t} \alpha e^{-\beta (t - t_i)},
\end{equation*}
where $\lambda^\infty$ denotes the background rate that the intensity 
reverts to over time, 
$\alpha$ denotes the jump size for the intensity function, 
$\beta$ denotes the exponential decay rate, and 
the $t_i$'s denote times of events that occurred prior to time $t$. 
We refer to this model as the \emph{block Hawkes model} (BHM).

\subsection{Relation to Other Models}
\label{sec:OtherModels}
The block Hawkes model we consider is a simpler 
version of the Hawkes IRM, which couples a non-parametric Bayesian version of the SBM called the infinite relational model (IRM) \citep{Kemp2006} with mutually-exciting Hawkes processes. 
By utilizing mutually-exciting Hawkes processes, the Hawkes IRM allows for reciprocity of events between block pairs. 
Similar to the BHM, node pairs in a block pair are selected at random to form an edge. 
The authors use MCMC-based inference that scales only to very small networks. 

The BHM simplifies the Hawkes IRM by 
using a fixed number of classes $K$ and univariate rather than multivariate 
Hawkes processes. 
The use of univariate Hawkes processes is crucial because it allows for 
independence between block pairs, which we used in the analysis in Section 
\ref{sec:RelationSBM} to demonstrate an asymptotic equivalence with an SBM. 
We use this asymptotic equivalence 
in Section \ref{sec:Inference} to devise an efficient 
inference procedure that scales to networks with thousands of nodes.

The BHM also has similarities with the relational event model (REM) 
\citep{Dubois2013}, which associates a non-homogeneous Poisson process with each pair of nodes, where the intensity function is piecewise constant with knots (change points) at the event times. 
Different node pairs belonging to the same block pair are governed by the same set of parameters. 
The REM also incorporates other edge formation mechanisms within block pairs such as reciprocity and transitivity, similar to an exponential random graph model. 
The authors also use MCMC for inference.

%% file: inference.tex
\section{Inference Procedure}
\label{sec:Inference}
The observed data is in the form of triplets $\vec{e}_s = (u_s, v_s, t_s)$ for each event $s$ denoting the nodes $u_s, v_s$ involved and the timestamp $t_s$. 
Consider an event matrix $E$ where each row corresponds to an event in the form of a triplet $\vec{e}_s$. 
Fitting the BPPM involves estimating both the unknown classes or blocks for each node and the point process parameters $\theta_b$ for each block pair $b$ from $E$. 
In the case of an exponential Hawkes process, the parameters are given by
\begin{equation*}
\theta_b = \left(\alpha_b, \beta_b, \lambda_b^{\infty}\right)
\end{equation*}
for each block pair $b$.
Let $\theta = \{\theta_b\}_{b=1}^p$ denote the set of point process parameters over all $p$ block pairs. 

Exact inference is impractical for all but the smallest networks due to the discrete class memberships $\vec{c}$. 
Thus, we consider two approximate inference methods: a greedy local search (Section \ref{sec:LocalSearch}) and variational inference (Section \ref{sec:Variational}). 
Both approaches are iterative and converge to a local maximum and are thus sensitive to the choice of initialization, which we discuss in Section \ref{sec:Initial}.

\subsection{Local Search}
\label{sec:LocalSearch}
Consider a conditional likelihood function for the point process parameters $\theta$ given the values of the class memberships $\vec{c}$ that determine the block pairs. 
Let $E^{(b)}$ denote rows of $E$ corresponding to events involving block pair 
$b = (b_1, b_2)$; that is, rows $\vec{e}_s$ where $u_s \in b_1$ 
and $v_s \in b_2$. 
The $p$ row blocks $E^{(b)} = [\vec{u}^{(b)}, \vec{v}^{(b)}, \vec{t}^{(b)}]$ form a partition of the rows of matrix $E$. 
Let $m_b$ denote the number of events observed in block pair $b$. 
Let $n_b$ denote the size of block pair $b$, i.e.~the number of node pairs or possible 
edges in block pair $b$, which is given by 
$|b_1|(|b_1|-1)$ if $b_1 = b_2$ and 
$|b_1| |b_2|$ otherwise.
The conditional log-likelihood function is given by
\begin{align}
\log \Pr(E | \theta, \vec{c}) &= \log \prod_{b=1}^p \Pr\left(\vec{u}^{(b)}, \vec{v}^{(b)}, 
	\vec{t}^{(b)} \Big| \theta_b, \vec{c}\right) \nonumber \\
&= \log \prod_{b=1}^p \Pr\left(\vec{t}^{(b)} \big| \theta_b \right) \prod_{s \in b} 
	\Pr(u_s, v_s | \vec{c}) \nonumber \\
&= \log \prod_{b=1}^p \Pr\left(\vec{t}^{(b)} \big| \theta_b \right) \left(\frac{1}{n_b}\right)^{m_b} \nonumber \\
\label{eq:LogLikNoPP}
&= \sum_{b=1}^p \left[\log \Pr\left(\vec{t}^{(b)} \big| \theta_b \right) - m_b 
	\log n_b\right]
\end{align}
where the expression for $\Pr(u_s, v_s | \vec{c})$ follows from the random selection of nodes in 
step \ref{item:selectNodes} of the BPPM generative process. 
The term $\log \Pr\big(\vec{t}^{(b)} | \theta_b \big)$ is simply the log-likelihood of 
the point process model parameters given the timestamps 
of events in block pair $b$. 
For the block Hawkes model, this term can be expressed in the following form \cite{ozaki79,Laub2015}:
\begin{align}
\log \Pr&\left(t_{(1)}, \ldots, t_{(m)} | \alpha_b, \beta_b, \lambda_b^\infty\right) = \sum_{s=1}^{m} \Bigg\{\frac{\alpha_b}{\beta_b} 
	\Big[e^{-\beta_b \left(t_{(m)} - t_{(s)}\right)} - 1\Big] 
 \nonumber \\
\label{eq:HawkesLogLik}
&+ \log\left[\lambda_b^\infty + \alpha_b \sum_{r=1}^{s-1} e^{-\beta_b \left(t_{(s)} - t_{(r)}\right)}\right]\Bigg\} -\lambda_b^\infty t_{(m)},
\end{align}
where
$t_{(s)}$ denotes the $s$th event corresponding to block pair $b$. 
\eqref{eq:HawkesLogLik} can be written in a recursive form as shown in \cite{ogata1981}.

The conditional log-likelihood \eqref{eq:LogLikNoPP} requires knowledge of class memberships $\vec{c}$, which are used to partition the event matrix into 
row blocks $E^{(b)}$ and thus affect both terms in \eqref{eq:LogLikNoPP} 
through $\vec{t}^{(b)}$, $m_b$, and $n_b$. 
However, in practice, class memberships are unknown, so we must maximize \eqref{eq:LogLikNoPP} over all possible class assignments. 

We use a local search (hill climbing) procedure, which is also often referred 
to as label switching or node swapping in the network science literature 
\citep{Karrer2011,Zhao2012} to iteratively update the class 
assignments to reach a local maximum in a greedy fashion. 
Recent work has found that such greedy algorithms are competitive with more 
computationally demanding estimation algorithms in both the static SBM 
\citep{Come2015} and discrete-time dynamic SBM \citep{Xu2014a,Corneli2015} 
while scaling to much larger networks. 
At each iteration, we swap a single node to a different class by choosing 
the swap that increases the log-likelihood the most. 
For each possible swap, we evaluate the log-likelihood by partitioning events 
using to the new class assignments, obtaining the 
maximum-likelihood estimates of the point process model parameters, and 
substituting these estimates along with the 
new class assignments into \eqref{eq:LogLikNoPP}. 
For the block Hawkes model, we maximize \eqref{eq:HawkesLogLik} with respect to 
$(\alpha_b, \beta_b, \lambda_b^\infty)$ for each block using a 
standard interior point optimization routine \cite{byrd2000}.

Each iteration of the local search considers $N(K-1)$ possible swaps. 
Computing the log-likelihood for each swap involves iterating over the 
timestamps of all $M$ events. 
Thus, each iteration of the local search has time complexity $O(KMN)$, 
which is linear in both the number of nodes and events, allowing 
it to scale to large networks. 
We verify this time complexity experimentally in Section \ref{sec:ExpScal}. 
The local search is easily parallelized by evaluating each possible 
swap on a separate CPU core. 
We terminate the local search procedure when no swap is able to increase 
the log-likelihood, indicating that we have reached a local maximum.

\subsection{Variational Inference}
\label{sec:Variational}
Variational inference is commonly used as an optimization-based alternative to MCMC and scales to much larger data sets. 
We implement a mean-field variational inference approach for the BPPM to approximate the intractable posterior distribution by a fully factorizable variational distribution. 
To reduce the Kullback-Leibler (KL) divergence between the true posterior and the mean-field approximation, we derive the evidence lower bound (ELBO) using an approach similar to the derivation in \cite{daudin2008} for a static SBM. 

Let $Z$ denote an $N \times K$ class membership matrix, where the notation $z_{iq} = 1$ is equivalent to $c_i = q$, both denoting that node $i$ is in class $q$. 
We use a $K$-dimensional multinomial distribution for each row $\vec{z}_i = [z_{i1},z_{i2},\ldots,z_{iK}]$ of $Z$ resulting in the following variational distribution:
\begin{equation}
\label{eq:VarDist}
R_E(Z) = \prod_{i=1}^{N} \text{Multinomial}(\vec{z}_i|\vec{\tau}_i),
\end{equation}
where $\vec{\tau}_i$ denotes the variational parameter for node $i$. 
Unlike for a static SBM, we don't have a closed-form update equation for the block Hawkes model, so we optimize the ELBO using coordinate ascent. 
The derivation of the ELBO and the variational expectation-maximization algorithm are provided in Appendix \ref{sec:AppendixInference}.

\subsection{Spectral Clustering Initialization}
\label{sec:Initial}
In order to ensure that the local search or variational inference procedures do not get stuck in poor 
local maxima, it is important to provide a good initialization. 
Methods used to initialize class estimates in static and discrete-time SBMs 
include variants of k-means clustering \citep{Come2015,Matias2016,matias2015semiparametric} and spectral 
clustering \citep{Xu2014a,Xu2015}. 
Variational inference is often executed with multiple random initializations, although some structured approaches have been used successfully in practice for certain models.
Given the close relationship between the proposed BPPM and the SBM discussed 
in Section \ref{sec:RelationSBM}, we use a spectral clustering initialization, which is much faster and more principled than the 
typical approach of multiple random initializations. 
Spectral clustering is an attractive choice because it scales to 
large networks containing thousands of nodes and has
theoretical performance guarantees applicable to the BPPM, as we discuss next. 

Recent work has demonstrated that applying spectral clustering (or a 
regularized variant) to a network generated from an SBM  
results in consistent estimates of class assignments as the number of nodes 
$N \rightarrow \infty$ \citep{Rohe2011,Sussman2012,qin13,Lei2015,rohe2016}. 
These theoretical guarantees typically require the expected degrees of nodes 
to grow polylogarithmically with the number of nodes so that the network is 
not too sparse. 
Networks that satisfy this requirement belong to the polylog degree regime. 
On the other hand, the \nameref{thm:AsyIndep} shows an asymptotic 
equivalence between the BPPM and SBM provided that there are not too many 
events, i.e.~the network is not too dense. 

\begin{algorithm}[t]
\caption{Regularized spectral clustering algorithm used to initialize the 
	local search in the BPPM inference procedure}
\label{alg:SpectralClust}

\begin{algorithmic}[1]
\REQUIRE Adjacency matrix $A$, number of classes $K$, regularization parameter 
$\tau \geq 0$ (Default: $\tau$ = average node degree)
\STATE Compute diagonal matrices $O^{\tau}$ with entries $o_{ii}^\tau = \sum_j a_{ij} + \tau$ and $P^{\tau}$ with entries $p_{jj}^\tau = \sum_i a_{ij} + \tau$
\STATE Compute regularized graph Laplacian
$L = (O^\tau)^{-1/2}A(P^\tau)^{-1/2}$
\STATE Compute singular value decomposition of $L$
\STATE $\tilde{\Sigma} \leftarrow$ diagonal matrix of $K$ largest 
	singular values of $L$
\STATE $(\tilde{U}, \tilde{V}) \leftarrow$ left and right 
	singular vectors for $\tilde{\Sigma}$
\STATE $\tilde{Z} \leftarrow [\tilde{U}, 
	\tilde{V}]$ \COMMENT{concatenate 
	left and right singular vectors}
\STATE Normalize each row of $\tilde{Z}$ to have magnitude of $1$
\STATE $\hat{\vec{c}} \leftarrow$ k-means clustering on rows of 
	$\tilde{Z}$
\RETURN $\hat{\vec{c}}$
\end{algorithmic}
\end{algorithm}

In the polylog degree regime, the ratio
\begin{equation*}
\frac{\mu_b}{n_b} = O\left(\frac{N \poly(\log N)}{N^2}\right) \rightarrow 0 
	\text{ as } N \rightarrow \infty,
\end{equation*}
so the network is not too dense, and the \nameref{thm:AsyIndep} holds. 
Thus, spectral clustering should provide an accurate estimate of the class 
assignments in the polylog degree regime, which is commonly observed in real 
networks such as social networks. 
Since we consider directed relations, we use a regularized spectral clustering 
algorithm for directed networks (pseudocode provided in 
Algorithm \ref{alg:SpectralClust}) to initialize the local search. 
It is a variant of the DI-SIM co-clustering algorithm \citep{rohe2016} 
modified to produce a single set of clusters for directed graphs by concatenating scaled left and right singular vectors in a manner 
similar to \citet{Sussman2012}.
For variational inference, we require an initialization on the variational parameters $\vec{\tau}_i$ rather than a hard clustering solution, so we set $\vec{\tau}_i$ to the $i$th row of $\tilde{Z}$ to represent a soft clustering initialization.

%% file: experiments.tex
\section{Simulated Data Experiments}

\subsection{Deviation from Independence}
\label{sec:DevExp}
The \nameref{thm:AsyIndep} demonstrates that pairs of adjacency matrix 
entries in the same block pair are dependent, but that the dependence is upper 
bounded by \eqref{eq:deltaBound}, and that the dependence goes to $0$ for 
growing blocks. 
To evaluate the slackness of the bounds, we simulate networks from the block 
Hawkes model (BHM). 
Since $\delta_0$ and $\delta_1$ depend only on the size of the blocks, 
we simulate networks with 
a single block and let the number of nodes $N$ 
grow from $10$ to $1,000$. 
For each number of nodes, we simulate $100,000$ networks from the block 
Hawkes model for a duration of $T = 20$ time units. 
We choose the Hawkes process parameters to be $\alpha = 5N$, $\beta = 10N$, 
and $\lambda^\infty = 0.5N$. 
The expected number of events $\mu = NT = 20N$, which 
grows with $N$ and is slower than the growth of the size $n = N(N-1)$ 
of the block pair, so the \nameref{thm:AsyIndep} applies.

\begin{figure}[t]
\centering
\subfloat{\includegraphics[width=1.6in]{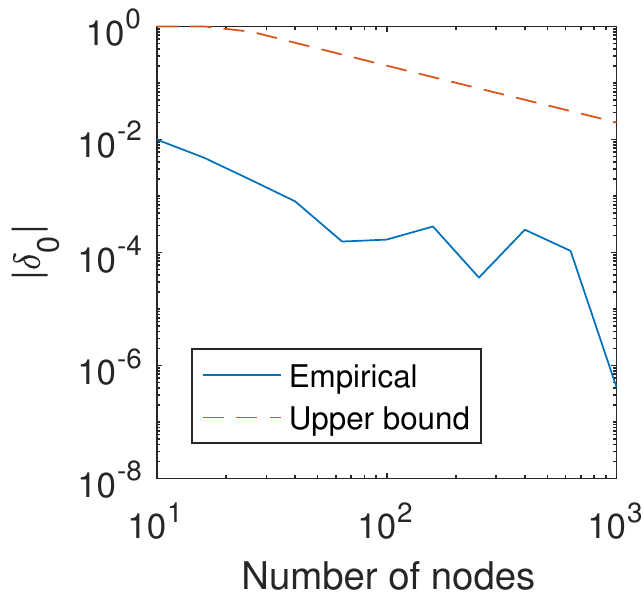}}
\quad
\subfloat{\includegraphics[width=1.6in]{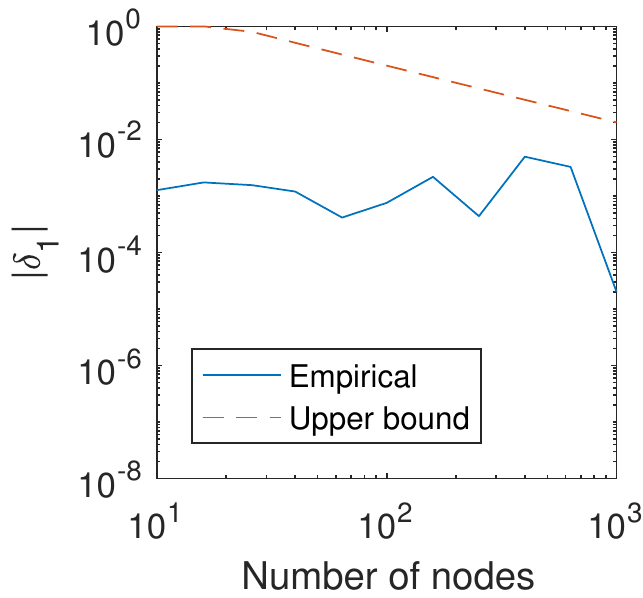}}
\caption{Comparison of empirical deviation from independence with theoretical 
upper bound. 
Empirical deviations are well within the theoretical bound.}
\label{fig:devIndep}
\end{figure}

We evaluate the absolute difference 
between the empirical marginal probability 
$\widehat{\Pr}(a_{ij}=0)$ and the empirical conditional probabilities 
$\widehat{\Pr}(a_{ij}=0 | a_{i'j'}=0)$ and 
$\widehat{\Pr}(a_{ij}=0 | a_{i'j'}=1)$. 
The empirical deviation from independence is shown to be well below the 
upper bound in Figure \ref{fig:devIndep}. 
The bound \eqref{eq:deltaBound} in the \nameref{thm:AsyIndep} depends 
only on the mean 
number of events, so it is somewhat loose when applied to the block Hawkes model. 

\subsection{Class Estimation}
\label{sec:ClassEstSim}

\begin{figure}[t]
\centering
\includegraphics[width = 3.3in]{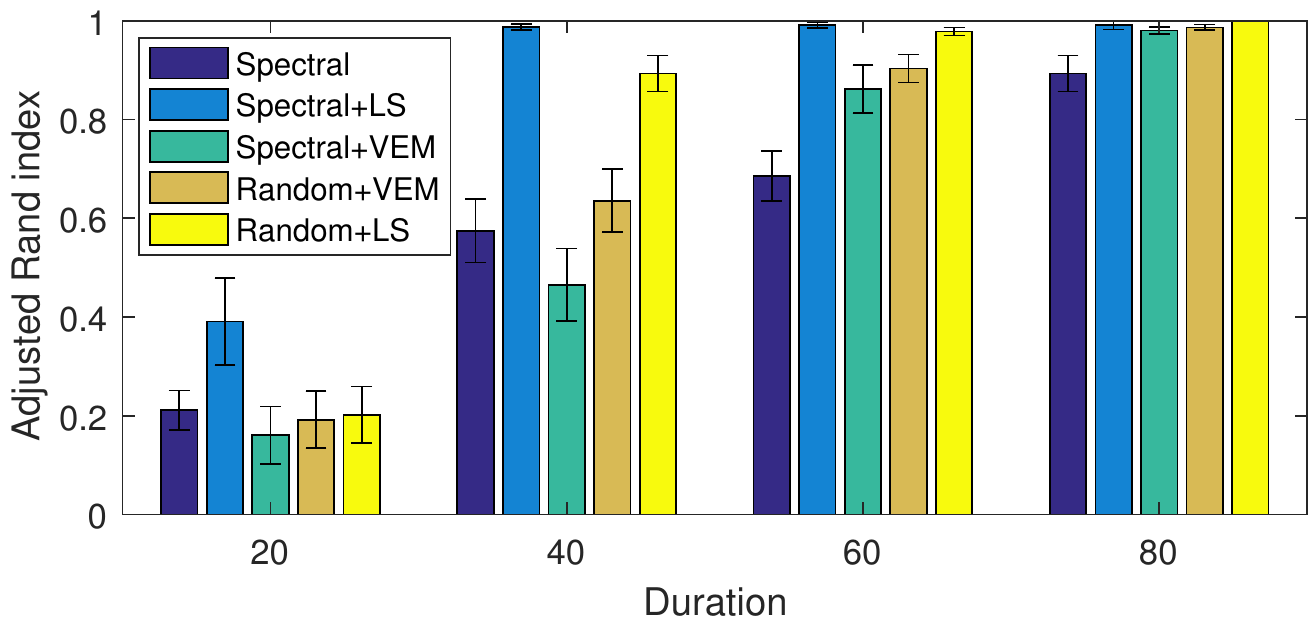}
\caption{Mean adjusted Rand indices ($\pm$ standard error) for class estimation simulation experiment with varying durations. 
Spectral clustering with local search produces the most accurate class estimates, while the spectral initialization for variational EM does not appear to work well.}
\label{fig:SimRandIndex}
\end{figure}

This simulation experiment is based on the synthetic network generator from \citet{Newman2004}, where all the diagonal block pairs have the same parameters, and the off-diagonal block pairs have the same parameters, but different from the diagonal block pairs.
We generate networks with $128$ nodes and $4$ classes from the block Hawkes 
model using Algorithm \ref{alg:GenBPPM} with varying durations from $20$ to 
$80$ time units.
We generate $10$ networks for each duration, with Hawkes process parameters $\alpha$ and $\beta$ being $0.6$ and $0.8$, respectively, for all block pairs. 
The baseline rates $\lambda^\infty$ are $1.8$ for diagonal block pairs and $0.6$ for off-diagonal block pairs. 
Classes are estimated using $5$ methods: spectral clustering, local search initialized with spectral clustering, variational EM initialized with spectral clustering, variational EM with $10$ random initializations, and local search with $10$ random initializations. 

The results shown in Figure \ref{fig:SimRandIndex} demonstrate that local search initialized with spectral clustering (Spectral+LS) is the most accurate. 
Conversely, the spectral clustering initialization for variational EM (Spectral+VEM) results in significantly worse results. 
Variational EM and local search using $10$ random initializations (Random+VEM and Random+LS) take significantly longer than Spectral+LS; however, they are also less accurate, so we use Spectral+LS in the remainder of this paper.

\subsection{Scalability of Local Search}
\label{sec:ExpScal}

\begin{figure}[t]
\centering
\subfloat{
\includegraphics[width = 1.5in]{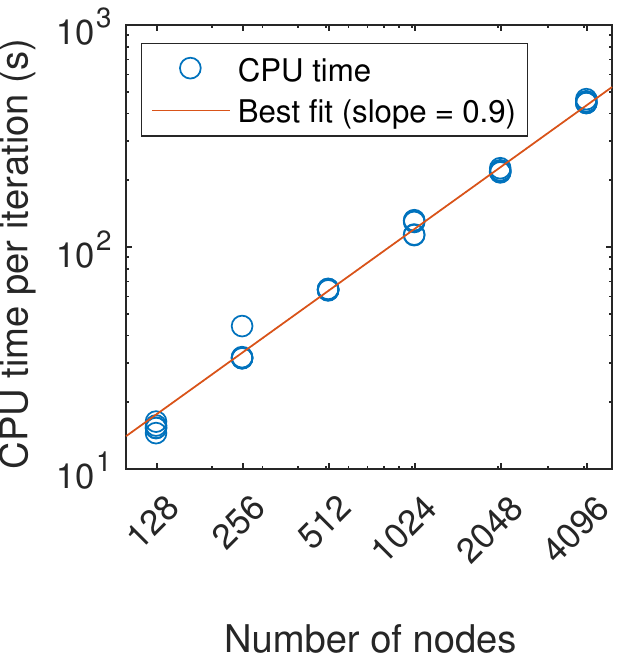}
}
\quad
\subfloat{
\includegraphics[width = 1.5in]{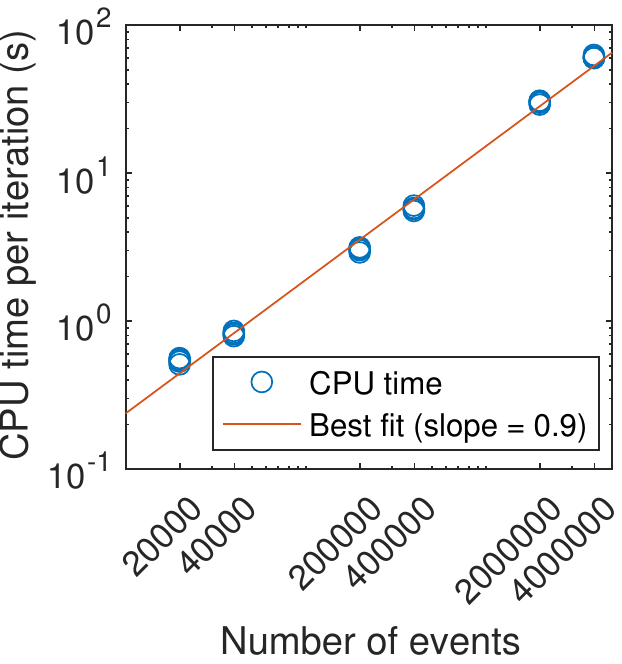}
}
\caption{CPU time per iteration of local search inference in seconds with varying number of nodes and events.}
\label{fig:Scalability}
\end{figure}

\begin{figure*}[t]
	\centering
	\subfloat[Singular values]
	{\includegraphics[height=1.2in]{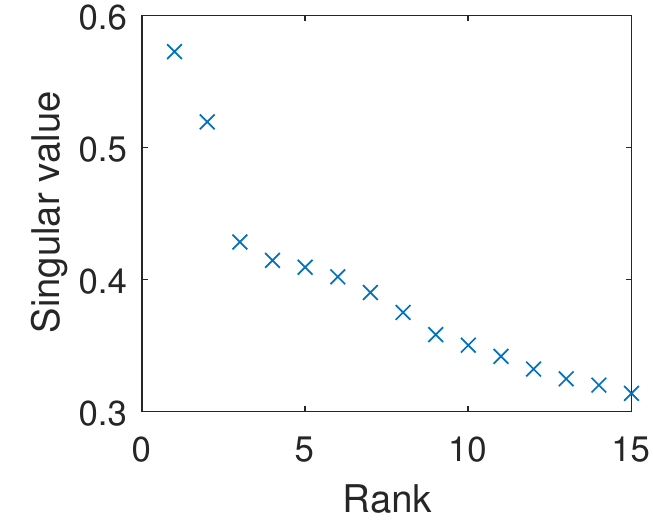}
	\label{fig:RealitySingularValues}}
	\quad
	\subfloat[$\alpha$: Jump sizes]
	{\includegraphics[height=1.2in]{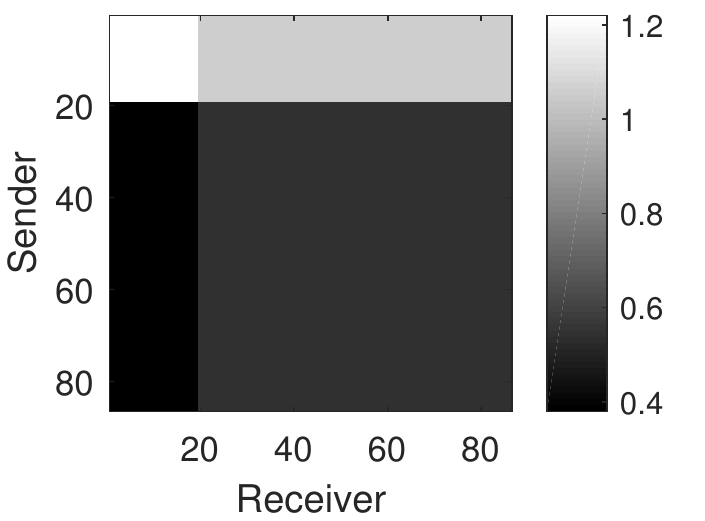}
	\label{fig:RealityAlpha}}
	\quad
	\subfloat[$\beta$: Jump decay rates]
	{\includegraphics[height=1.2in]{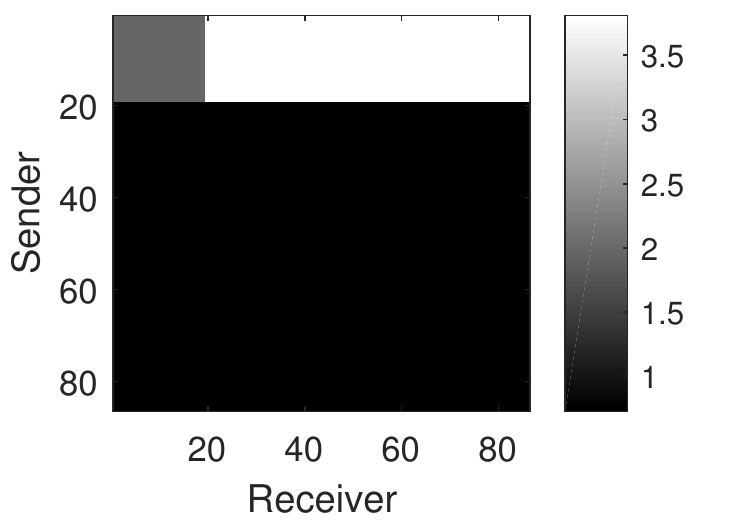}
	\label{fig:RealityBeta}}
	\quad
	\subfloat[$\lambda^\infty$: Background intensity]
	{\includegraphics[height=1.2in]{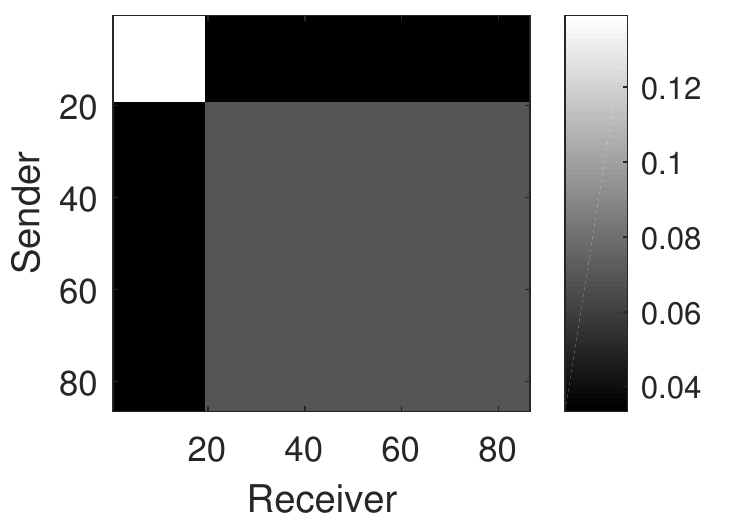}
	\label{fig:RealityLambda}}
	\caption[Block Hawkes model fit on Reality Mining data.]
	{Block Hawkes model fit on Reality Mining data with 2 blocks. 
    Nodes 1-15 belong to block 1 and the rest to block 2. \subref{fig:RealitySingularValues} First 15 singular values of the adjacency matrix. \subref{fig:RealityAlpha}-\subref{fig:RealityLambda} Hawkes process parameter estimates.}
	\label{fig:RealityMiningFit}
\end{figure*}

We evaluate the scalability of the proposed local search inference procedure by generating networks with varying number of nodes and events from the block Hawkes model. 
When varying the number of nodes, we choose a time duration of $1,200$ time units and set the Hawkes process parameters 
$(\alpha,\beta,\lambda^\infty)$ to $(1.6,2,1.2)$ for diagonal block pairs and $(0.6,0.8,0.6)$
for off-diagonal block pairs. 
When varying the number of events, we choose a network with $128$ nodes and set the Hawkes process parameters 
$(\alpha,\beta,\lambda^\infty)$ to $(0.6,0.8,1.8)$ for diagonal block pairs and $(0.6,0.8,0.6)$
for off-diagonal block pairs. 
We keep the number of classes fixed to be $4$ in both settings and simulate $5$ networks for each configuration.

The CPU times per iteration for both varying number of nodes and events are shown in Figure \ref{fig:Scalability} along with 
best-fit lines for a power law relationship (beginning with $512$ nodes for the varying nodes experiment). 
The best-fit line has slope $0.9$ in both cases, confirming the 
linear time complexity both in terms of the number of nodes $N$ and the number of events $M$, which is as expected according to the computed time complexity 
in Section \ref{sec:LocalSearch}. 
All CPU times are recorded on a Linux workstation using $18$ Intel Xeon 
processor cores operating at $2.8$ GHz.

\section{Real Data Experiments}

\subsection{MIT Reality Mining}
We analyze the MIT Reality Mining data \cite{eagle09}, using start times of phone calls as events from the caller to the recipient. 
Nodes in this data set correspond to students and staff at MIT. 
The data were collected over approximately 10 months beginning in August 2004.
We remove all nodes who do not either make or receive a call at any point in the data trace, resulting in a network with 86 nodes.

We observe a significant gap after the second largest singular value of the regularized graph Laplacian, as shown in Figure \ref{fig:RealitySingularValues} so we choose $K=2$ blocks. 
From examining the block Hawkes model parameter estimates shown in Figure \ref{fig:RealityAlpha}-\ref{fig:RealityLambda}, block pair $(1,1)$ has both higher background intensity $\lambda^\infty$ and longer bursts due to high $\alpha$ and low $\beta$. 
Thus, not only does it appear to be a community, but phone calls within the community tend to happen in prolonged bursts. 
On the other hand, block pair $(2,2)$ has higher background intensity than the off-diagonal block pairs but without the large jump sizes, indicating a lack of bursty behavior compared to block pair $(1,1)$. 
Block pair $(1,2)$ has high values for both $\alpha$ and $\beta$, indicating large bursts of short duration. 
However, since block pair $(1,2)$ has a lower background intensity than $(2,2)$, the overall density of the block pair in the aggregated adjacency matrix is not necessarily higher. 
Indeed, block pair $(2,2)$ has a density of $0.041$, while block pair $(2,1)$ has a density of $0.034$. 
Thus, the continuous-time model of the network enables greater analysis of the dynamics of interactions over both short and long periods of time, and the findings may be quite different from those of static and discrete-time network representations. 

\subsection{Facebook Wall Posts}
\label{sec:ExploratoryAnalysis}
\subsubsection{Model-Based Exploratory Analysis}

\begin{figure*}[t]
	\centering
	\subfloat[Singular values]
	{\includegraphics[height=1.2in]{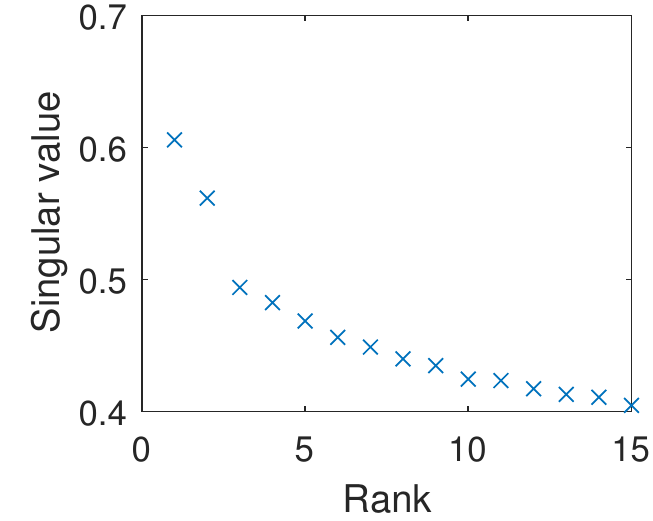}
	\label{fig:FacebookSingularValues}}
	\quad
	\subfloat[$\alpha$: Jump sizes]
	{\includegraphics[height=1.2in]{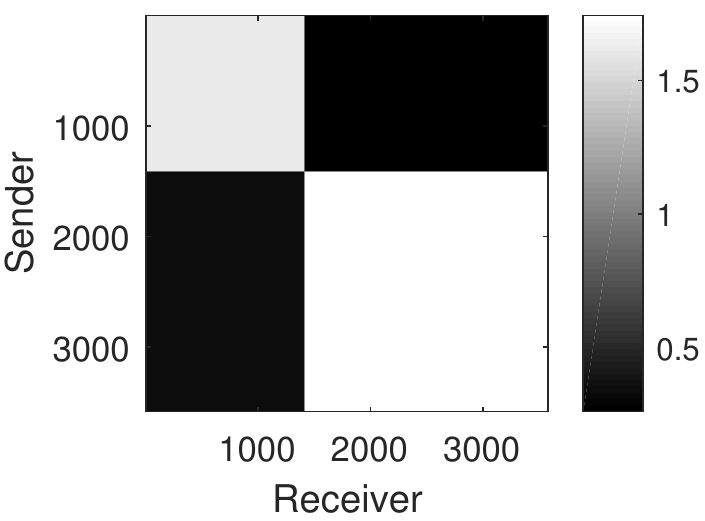}
	\label{fig:FacebookAlpha}}
	\quad
	\subfloat[$\beta$: Jump decay rates]
	{\includegraphics[height=1.2in]{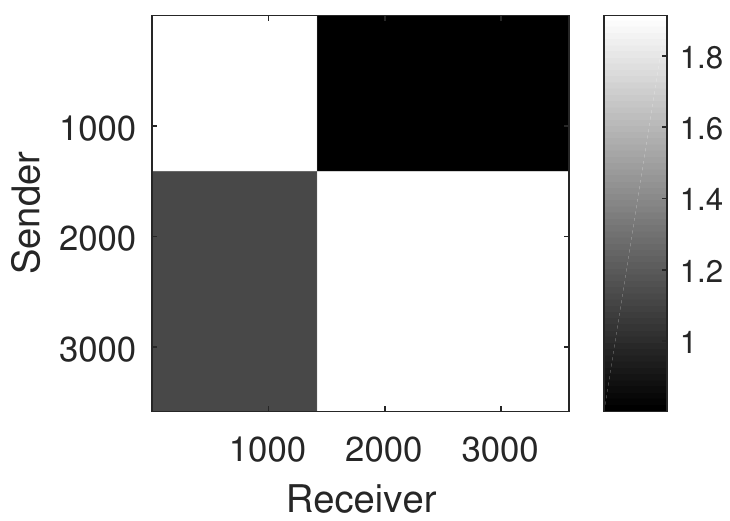}
	\label{fig:FacebookBeta}}
	\quad
	\subfloat[$\lambda^\infty$: Background intensity]
	{\includegraphics[height=1.2in]{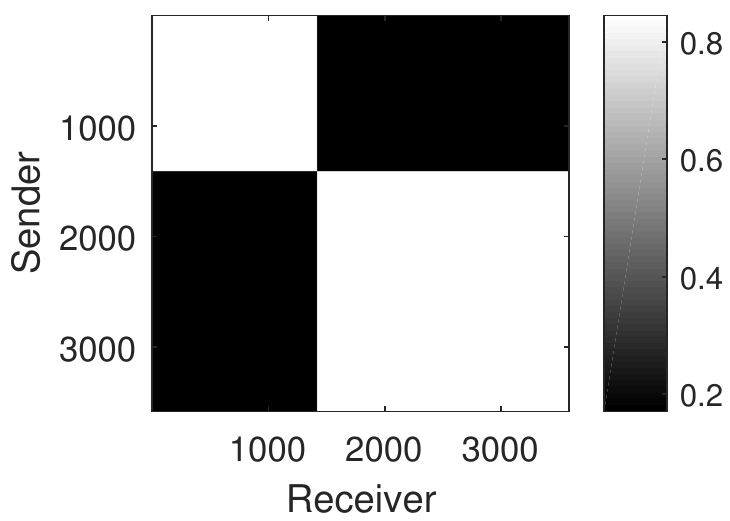}
	\label{fig:FacebookLambda}}
	\caption[Block Hawkes model parameter estimates on Facebook wall post data.]
	{Block Hawkes model parameter estimates on Facebook wall post data with 2 blocks. 
	Nodes 1-1421 belong to block 1 and the rest to block 2. 
	\subref{fig:FacebookSingularValues} First 15 singular values of the adjacency matrix. \subref{fig:FacebookAlpha}-\subref{fig:FacebookLambda} Hawkes process parameter estimates.}
	\label{fig:FacebookParams}
\end{figure*}

We analyze the Facebook wall post data collected 
by \citet{viswanath2009evolution}, which contains over $60,000$ nodes. 
We consider events between January 1, 2007 and January 1, 2008. 
We remove nodes with degree less than $10$, resulting in a network with 
$137,170$ events among $3,582$ nodes.

\begin{figure}[t]
\centering
\includegraphics[width=3.2in]{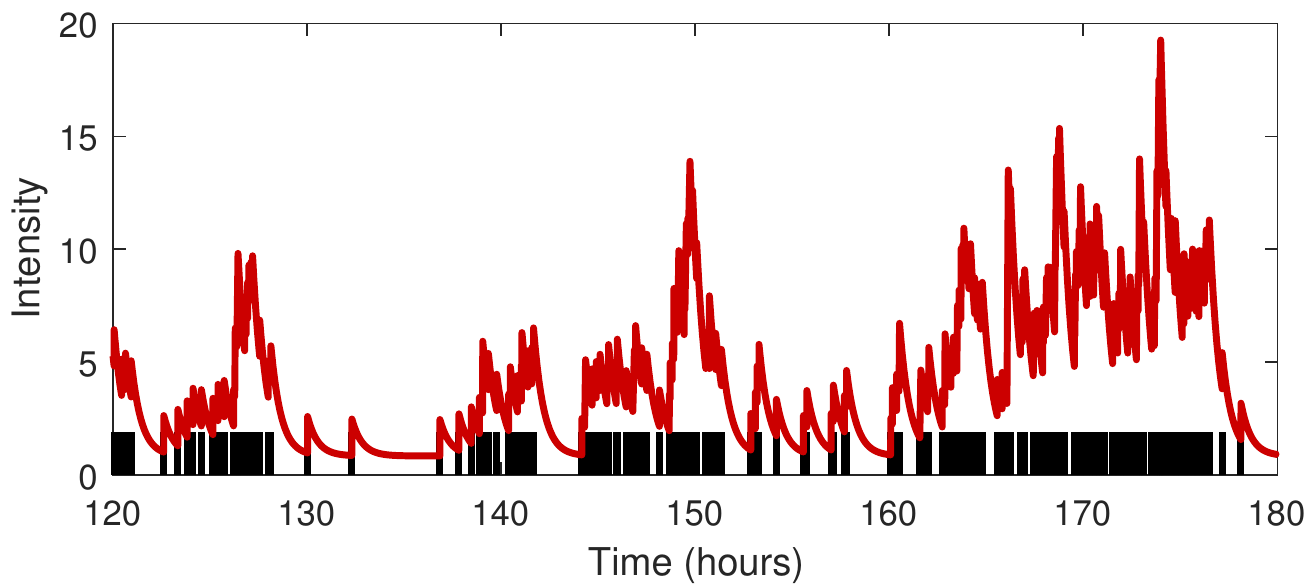}
\caption{Hawkes process intensity function for events in block pair (1,1) 
	of the Facebook data. 
Event times are denoted by black ticks along the horizontal axis.}
\label{fig:Intensity}
\end{figure}

We select a model with $K=2$ blocks, as suggested by the singular values of the regularized graph Laplacian shown in Figure \ref{fig:FacebookSingularValues}. 
The parameters inferred from the BHM fit on the Facebook data are shown in 
Figure \ref{fig:FacebookAlpha}-\ref{fig:FacebookLambda}. 
The diagonal block pairs have larger values of background intensity 
$\lambda^{\infty}$, indicating that the blocks form communities. 
This finding could also have been yielded by static and discrete-time SBMs.
The diagonal block pairs also have higher values of jump sizes $\alpha$, 
indicating that wall posts between members of a community are more bursty. 
A portion of the Hawkes process intensity function for block pair $(1,1)$ is 
shown in Figure \ref{fig:Intensity}. 
In addition to diurnal patterns, one can observe bursty periods of wall posts 
throughout the day. 
This finding could not have been obtained from static and discrete-time SBMs. 
From the values of $\alpha$ on the diagonal block pairs, we see that wall posts within block 2 are more bursty, with higher jump sizes and roughly equal jump decay rates compared to posts within block 1. 
By observing the values of $\alpha$ on off-diagonal block pairs, we notice 
that there isn't much asymmetry, but the decay rate $\beta$ 
exhibits asymmetry. 
Specifically, events from block 1 to 2 have longer sustained 
bursts than events from block 2 to 1 due to the lower value of $\beta_{12}$ compared to $\beta_{21}$.

\subsubsection{Comparison with Discrete-Time SBM}
\label{sec:ExpDisc}
To compare our proposed continuous-time BHM with a discrete-time SBM 
\citep{Xu2014a}, we consider the task of predicting the time to the next event in each block pair. 
We believe that this is a fair comparison because both the BHM and 
discrete-time SBM require that pairs of nodes in the same block pair have 
identical edge probabilities. 
We split the data into two sets, with the first $8$ months of data for training and the last $4$ months for testing.
During each week of the test set, we attempt to predict the time to the next event in each block pair. 
Afterwards, we update the model with all events during that week. 
This results in $16$ predictions in the test set for each block pair.

The BHM directly models event times, so we use the expected next
event time for each block as the prediction. 
The discrete-time SBM does not directly model event times, so we multiply the 
expected number of time snapshots that will elapse before the next edge 
formation, which is geometrically distributed, by the snapshot length and then subtract half the snapshot length (to center the prediction within the snapshot) to get a 
prediction for the next event time. 
Since the prediction for the discrete-time SBM is dependent on the snapshot 
length, we test several different snapshot lengths. 

We estimate class assignments 
for both models using regularized spectral clustering with $2$ classes (with no local 
search iterations) so that differences in class 
estimates between the two models do not play a role in the accuracy. 
We believe this is a valid comparison because spectral clustering 
is used as the initialization to local search 
in the inference procedure for both models, as discussed in Section 
\ref{sec:LocalSearch} for the BHM and in \cite{Xu2014a} 
for the discrete-time SBM.

We evaluate the accuracy of the predictions by computing the root mean-squared 
error (RMSE) between predicted event times and actual event times for 
the first event in each block pair during a week. 
Since the blocks form communities, we expect events to arrive much more frequently within blocks. 
Thus we separate the evaluation into within-block and between-block prediction RMSE.
As shown in Figure \ref{fig:Comparisons}, the accuracy of the discrete-time SBM 
is highly dependent on the snapshot length. 
For snapshots of 6 hours and longer, the loss in temporal resolution is the main 
contributor to the high RMSE. 
The shorter snapshots such as 1 hour and 2 hours have excellent temporal resolution for within-block prediction but are less accurate for between-block prediction. 
Choosing 3-hour long snapshots results in the most accurate between-block predictions. 
Due to the different rates of events within and between blocks, the discrete-time representation must trade off between within-block and between-block prediction accuracy when choosing the snapshot length. 
Using our continuous-time BHM, we avoid this complex 
problem of choosing the snapshot length and produce more accurate predictions (in total RMSE)
than the discrete-time model for any snapshot length, as shown in Figure 
\ref{fig:Comparisons}.

\begin{figure}[t]
\centering
\includegraphics[width = 2.7in]{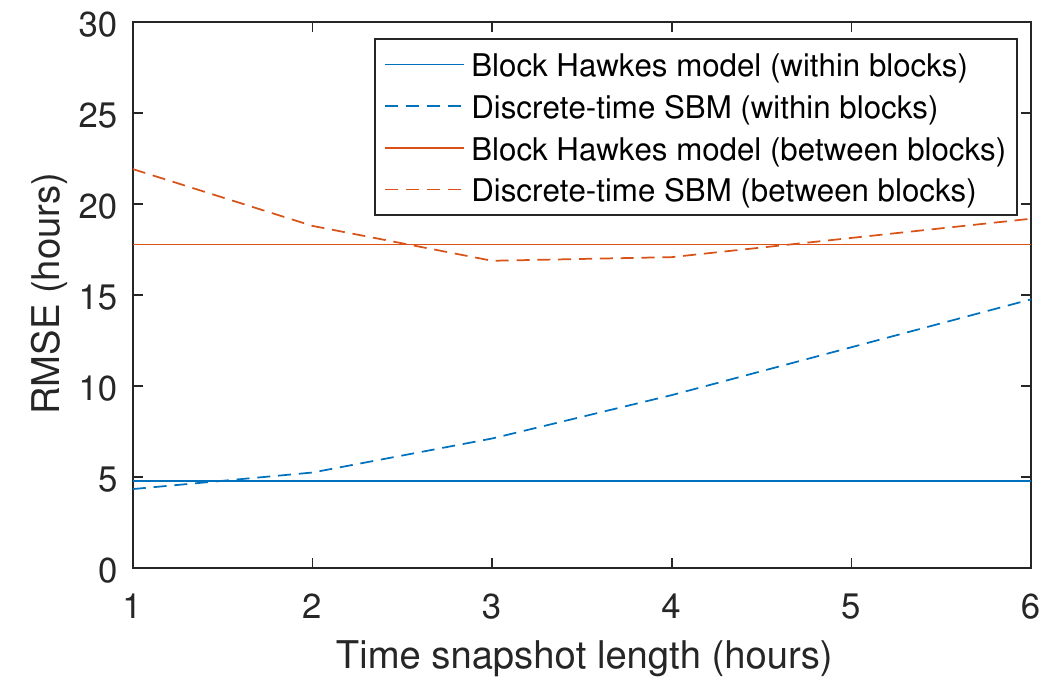}
\caption{Prediction RMSE in hours for block Hawkes model and discrete-time SBMs 
on Facebook data.}
\label{fig:Comparisons}
\end{figure}

%% file: conclusion.tex
\section{Conclusion}
\label{sec:Conclusion}

In this paper, we introduced the block point process model (BPPM) for 
dynamic networks evolving in continuous time in the form of timestamped 
events between nodes. 
Our model was inspired by the well-known stochastic block model (SBM) for 
static networks and is a simpler version of the Hawkes IRM. 
We demonstrated that adjacency matrices constructed from the BPPM follow 
an SBM in the limit of a growing number of nodes. 
To the best of our knowledge, this is the first result of this type connecting 
point process network models with adjacency matrix network models. 
Additionally we proposed a principled and efficient algorithm to fit the BPPM 
using spectral clustering and local search that scales to large networks and apply it to analyze several real networks.

%% file: proofDeviation.tex
\section{Proof of Asymptotic Independence Theorem}
\label{sec:proofDeviation}

We begin with a well-known lemma on the difference of powers that will be used 
both to upper and lower bound the deviation from independence.

\begin{lemma}[Difference of powers]
\label{lem:DiffPowers}
For a real number $x > 1$ and integer $m \geq 1$, we have the following 
identity:
\begin{equation}
m(x-1)^{m-1} \leq x^m - (x-1)^m \leq mx^{m-1}.
\label{eq:DiffPowers}
\end{equation}
\end{lemma}

\begin{proof}
The proof follows straightforwardly from factorizing a difference of powers. 
Specifically, for real numbers $x > y > 0$ and integer $m \geq 1$, 
\begin{equation}
x^m - y^m = (x-y) \sum_{i=0}^{m-1} x^{m-1-i} y^{i}.
\label{eq:Factorize}
\end{equation}
If $x > 1$ and $y = x-1$, then \eqref{eq:Factorize} becomes
\begin{equation*}
x^m - (x-1)^m = \sum_{i=0}^{m-1} x^{m-1-i} (x-1)^{i}.
\end{equation*}
There are $m$ terms in the summation. 
The largest term is $x^{m-1}$, and the smallest term is $(x-1)^{m-1}$. 
Thus, we can upper and lower bound the sum by $mx^{m-1}$ and $m(x-1)^{m-1}$, 
respectively, to arrive at \eqref{eq:DiffPowers}.
\end{proof}

The next lemma will be used in the upper bound.
\begin{lemma}
\label{lem:UBLemma}
For a real number $x > 1$ and integer $m \geq 1$,
\begin{equation*}
\left(\frac{n-1}{n}\right)^m \geq 1 - \frac{m}{n}.
\end{equation*}
\end{lemma}

\begin{proof}
\begin{align}
\left(\frac{n-1}{n}\right)^m - 1 &= \frac{(n-1)^m - n^m}{n^m} \nonumber \\
&= -\frac{n^m - (n-1)^m}{n^m} \nonumber \\
\label{eq:UBLemma1}
&\geq -\frac{mn^{m-1}}{n^m} \\
\label{eq:UBLemma2}
&= -\frac{m}{n},
\end{align}
where \eqref{eq:UBLemma1} follows from applying Lemma \ref{lem:DiffPowers}. 
Adding $1$ to both sides of \eqref{eq:UBLemma2}, we arrive at the desired 
result.
\end{proof}

We now prove the \nameref{thm:AsyIndep}. 

\begin{proof}[Proof of \nameref{thm:AsyIndep}]
First compute the marginal probability $\Pr(a_{ij}=0)$. 
$a_{ij}=0$ implies that no events between nodes $i$ and $j$ occurred. 
To compute this probability, we first compute the conditional probability 
given that the number of events in block pair $b$ is $m_b$. 
To simplify notation, we drop the subscript $b$ from $m_b$ and $n_b$ in the 
remainder of the proof, so the conditional probability can be written as 
\begin{equation*}
\Pr(a_{ij}=0|m) = \left(\frac{n-1}{n}\right)^m, \, m \geq 0, 
\end{equation*}
where the equality follows by noting that, conditioned on $m$ total 
events in block pair $b$, the number of events between nodes $i$ and $j$ 
follows a binomial distribution with $m$ trials and success probability $1/n$. 
The $1/n$ success probability is due to step \ref{item:selectNodes} of the 
generative process of the BPPM, which involves selecting node pairs randomly 
to receive an event. 
By the Law of Total Probability, the marginal probability is
\begin{align}
\Pr(a_{ij}=0) &= \sum_{m=0}^\infty p(m) \Pr(a_{ij}=0 | m) \nonumber \\
&= \sum_{m=0}^\infty p(m) \left(\frac{n-1}{n}\right)^m,
\label{eq:MargProb}
\end{align}
where the probability mass function $p(m)$ denotes the probability that 
$m$ events in block pair $b$ occurred.

Next consider the joint probability $\Pr(a_{ij}=0, a_{i'j'}=0)$. 
As before, condition on the number of events $m$. 
The conditional joint probability is
\begin{equation*}
\Pr(a_{ij} = 0, a_{i'j'} = 0|m) = \left(\frac{n-2}{n}\right)^m, \; m \geq 0,
\end{equation*}
because the number of events for each node pair in block pair $b$ 
follow a multinomial distribution 
with $m$ trials and all event probabilities equal to $1/n$. 
By the Law of Total Probability,
\begin{equation}
\label{eq:JointProb}
\Pr(a_{ij}=0,a_{i'j'}=0) = \sum_{m=0}^\infty p(m) \left(\frac{n-2}{n}
	\right)^{m}.
\end{equation}

We first lower bound $\delta_0$ by noting that
\begin{align}
\delta_0 &= \Pr(a_{ij}=0 | a_{i'j'}=0) - \Pr(a_{ij}=0) \nonumber \\
&\geq \Pr(a_{ij}=0,a_{i'j'}=0) - \Pr(a_{ij}=0) \nonumber \\
\label{eq:SumLowerBound}
&= \sum_{m=0}^\infty p(m) \left[\left(\frac{n-2}{n}\right)^m 
	- \left(\frac{n-1}{n}\right)^m\right] \\
&= -\sum_{m=0}^\infty p(m) \frac{(n-1)^m - (n-2)^m}{n^m} \nonumber \\
\label{eq:LowerBoundDiffPow}
&\geq -\sum_{m=0}^\infty p(m) \frac{m(n-1)^{m-1}}{n^m} \\
&= -\frac{1}{n} \sum_{m=0}^\infty m p(m) \left(\frac{n-1}{n}\right)^{m-1} 
	\nonumber \\
\label{eq:LowerBoundInterm}
&\geq -\frac{1}{n} \sum_{m=0}^\infty m p(m) \\
\label{eq:LowerBound}
&= -\frac{\mu}{n},
\end{align}
where \eqref{eq:SumLowerBound} follows from \eqref{eq:MargProb} and 
\eqref{eq:JointProb}, \eqref{eq:LowerBoundDiffPow} follows from Lemma 
\ref{lem:DiffPowers}, 
and \eqref{eq:LowerBoundInterm} follows by observing that 
$\big(\frac{n-1}{n}\big)^{m-1} \leq 1$. 

We now upper bound $\delta_0$ by noting that
\begin{align}
\delta_0 &= \Pr(a_{ij}=0 | a_{i'j'}=0) - \Pr(a_{ij}=0) \nonumber \\
&\leq 1 - \Pr(a_{ij}=0) \nonumber \\
&= 1 - \sum_{m=0}^\infty p(m) \left(\frac{n-1}{n}\right)^m \nonumber \\
\label{eq:UpperBoundDiffPow}
&\leq 1 - \sum_{m=0}^\infty p(m) \left(1 - \frac{m}{n}\right) \\
&= 1 - \sum_{m=0}^\infty p(m) + \frac{1}{n} \sum_{m=0}^\infty m p(m) 
	\nonumber \\
&= 1 - 1 + \frac{\mu}{n} \nonumber \\
\label{eq:UpperBound}
&= \frac{\mu}{n},
\end{align}
where \eqref{eq:UpperBoundDiffPow} follows from Lemma \ref{lem:UBLemma}. 

Next we lower bound $\delta_1$ by noting that
\begin{align}
\delta_1 &= \Pr(a_{ij}=0 | a_{i'j'}=1) - \Pr(a_{ij}=0) \nonumber \\
&= \frac{\Pr(a_{ij}=0, a_{i'j'}=1) - \Pr(a_{ij}=0) \Pr(a_{i'j'}=1)}
	{\Pr(a_{i'j'} = 1)} \nonumber \\
&\geq \Pr(a_{ij}=0, a_{i'j'}=1) - \Pr(a_{ij}=0) \Pr(a_{i'j'}=1) 
	\nonumber \\
&= \Pr(a_{ij}=0) - \Pr(a_{ij}=0, a_{i'j'}=0) \nonumber \\
&\qquad- \Pr(a_{ij}=0)[1 
	- \Pr(a_{i'j'}=0)] \nonumber \\
&= -[\Pr(a_{ij}=0, a_{i'j'}=0) - \Pr(a_{ij}=0) \Pr(a_{i'j'}=0)] \nonumber \\
&= -\Pr(a_{i'j'}=0) [\Pr(a_{ij}=0 | a_{i'j'}=0) - \Pr(a_{ij}=0)] \nonumber \\
&= -\Pr(a_{i'j'}=0) \delta_0 \nonumber \\
&\geq -\delta_0 \nonumber \\
\label{eq:LowerBound1}
&\geq -\frac{\mu}{n},
\end{align}
where \eqref{eq:LowerBound1} follows from \eqref{eq:UpperBound}.

Finally we upper bound $\delta_1$ using the same approach as for $\delta_0$: 
\begin{align}
\delta_1 &= \Pr(a_{ij}=0 | a_{i'j'}=1) - \Pr(a_{ij}=0) \nonumber \\
&\leq 1 - \Pr(a_{ij} = 0) \nonumber \\
\label{eq:UpperBound1}
&\leq \frac{\mu}{n},
\end{align}
where the final inequality is obtained from \eqref{eq:UpperBound}.

Combining \eqref{eq:LowerBound}, \eqref{eq:UpperBound}, 
\eqref{eq:LowerBound1}, and \eqref{eq:UpperBound1}, and noting that the 
maximum deviation between two probabilities is $1$, we arrive at the 
desired result $|\delta_0|,|\delta_1| \leq \min\{1,\mu/n\}$, which completes 
the proof.
\end{proof}

%% file: appendix.tex
\section{Variational Inference Details}
\label{sec:AppendixInference}
We begin by deriving the complete-data log-likelihood for the block point process model. 
\begin{align}
\log &\Pr(E,Z|\theta,\vec{\pi}) \nonumber \\
&= \log \Pr(Z|\vec{\pi}) + \log \Pr(\vec{u},\vec{v}, \vec{t}|Z,\theta) \nonumber \\
\label{eq:completeDataLogLikelihood}
&= \log \Pr(Z|\vec{\pi}) + \log \Pr(\vec{u},\vec{v}|Z) + \log \Pr(\vec{t}|\vec{u},\vec{v},Z,\theta)
\end{align}
The first term in \eqref{eq:completeDataLogLikelihood} represents the prior class probability and is given by 
\begin{equation*}
\log \Pr(Z|\vec{\pi}) = \log \left(\prod_{i=1}^N \prod_{q=1}^K \pi_{iq}^{z_{iq}}\right) = \sum_{i=1}^N \sum_{q=1}^K z_{iq}\log(\pi_{iq}).
\end{equation*}
The second term in \eqref{eq:completeDataLogLikelihood} represents the uniform distribution of events to nodes inside a block pair and is given by
\begin{align*}
\log \Pr(\vec{u},\vec{v}|Z) &= \log \left[\prod_{q=1}^K \prod_{l=1}^K \left(\frac{1}{n_{ql}}\right)^{m_{ql}}\right] \\
&= - \sum_{q=1}^{K}\sum_{l=1}^{K} m_{ql} \log(n_{ql}) \\
&= - \sum_{q=1}^{K}\sum_{l=1}^{K} \vec{z}_{:q}^{T} W \vec{z}_{:l} \log(\vec{z}_{:q}^{T} \vec{z}_{:l}),
\end{align*}
where $m_{ql}$ and $n_{ql}$ are defined in the same manner as $m_b$ and $n_b$ in Section \ref{sec:LocalSearch}. 
The last equality represents these two terms as a function of $Z$, where $\vec{z}_{:q}$ denotes the $q$th column of $\vec{z}$, and a weighted adjacency matrix $W$ where the weight $w_{ij}$ corresponds to the total number of events from node $i$ to node $j$. 
The third term in \eqref{eq:completeDataLogLikelihood} depends on the choice of point process; for the block Hawkes model, it is given by
\begin{align*}
\log &\Pr(\vec{t}|\vec{u},\vec{v},Z,\theta) =  \sum_{q=1}^{K}\sum_{l=1}^{K} \Biggl(\sum_{s=1}^M\Biggl\{z_{u_s q} z_{v_s l}  \frac{\alpha_{ql}}{\beta_{ql}}\left[e^{-\beta_{ql}(t_M-t_s)} - 1\right] \nonumber \\
& + \log \left[ \lambda_{ql}^{\infty} +\sum_{r=1}^{s-1} z_{u_r q} z_{v_r l} \alpha_{ql} e^{-\beta_{ql}(t_s-t_r)}\right]\Biggr\} - t_M\lambda_{ql}^{\infty} \Biggr),
\end{align*}
where $u_s$ and $v_s$ denote the sender and receiver nodes for event $s$.
The complete-data log-likelihood \eqref{eq:completeDataLogLikelihood} is intractable since changing the class membership of
one node might affect the class memberships of other nodes, so all possibilities
of $Z$ need to be considered.

We derive a mean-field variational inference procedure in which we approximate the posterior with the fully
factorizable variational distribution given by \eqref{eq:VarDist}. 
We now try to find the best distribution from the family of multinomial distributions
that can get us closest in KL divergence to the posterior. 
But calculating KL divergence requires calculating the intractable posterior, so we instead maximize the evidence lower bound (ELBO), which is equivalent to minimizing KL divergence \cite{blei2017}. 
The ELBO is given by
\begin{equation*}
\text{ELBO}(R_E) = \E\left[\log \Pr(E,Z|\theta,\vec{\pi})\right]-\E\left[\log R_E(Z)\right]. 
\end{equation*}
Now expand and calculate expectation of individual terms in the complete-data log-likelihood \eqref{eq:completeDataLogLikelihood}.
The expectation of the log of the summations in the second and third terms in \eqref{eq:completeDataLogLikelihood} can be bounded using Jensen's inequality: 
\begin{equation*}
\E [f(x)] \geq f(\E(x)).
\end{equation*}
After bounding the log sums using Jensen's inequality and calculating the
expectations with respect to $R_{E}$, we arrive at the simplified expression for the ELBO: 
\begin{align}
\mathcal{E}&(R_E) =\sum_{i=1}^N\sum_{q=1}^K\tau_{iq}\log(\pi_{iq}) - \sum_{q=1}^{K}\sum_{l=1}^{K} \vec{\tau}_{:q}^{T} W \vec{\tau}_{:l} \log(\vec{\tau}_{:q}^{T} \vec{\tau}_{:l}) \nonumber \\ 
&+ \sum_{q=1}^{K}\sum_{l=1}^{K} \Biggl(\sum_{s=1}^M\Biggl\{\tau_{u_s q} \tau_{v_s l} 
\frac{\alpha_{ql}}{\beta_{ql}}\left[e^{-\beta_{ql}(t_M-t_s)} - 1\right] + \log \Biggl[ \lambda_{ql}^{\infty} \nonumber \\
\label{eq:ELBO}
&+\sum_{r=1}^{s-1} \tau_{u_r q} \tau_{v_r l} \alpha_{ql} e^{-\beta_{ql}(t_s-t_r)}\Biggr]\Biggr\} - t_M\lambda_{ql}^{\infty} \Biggr) +\sum_{i=1}^N\sum_{q=1}^{K}\tau_{iq}\log(\tau_{iq})
\end{align}
where $\vec{\tau}_{:q}$ denotes the vector of variational parameters for class $q$ over all $N$ nodes. 
The variational expectation-maximization (VEM) algorithm alternates between a variational E step, in which we use coordinate ascent to optimize the ELBO \eqref{eq:ELBO} over the variational parameters $\vec{\tau}_i$ for each node $i$, and an M step, in which we optimize the ELBO over the Hawkes process parameters $\big(\alpha_{ql}, \beta_{ql}, \lambda_{ql}^{\infty}\big)$ for all block pairs $(q,l) \in \{1, \ldots, K\}^2$.